\documentclass[12pt]{article}

\usepackage{amsmath}
\usepackage{amssymb}
\usepackage{amsthm}
\usepackage{bm}
\usepackage{color}
\usepackage{fullpage}
\usepackage{graphicx}
\usepackage{natbib}
\usepackage{stmaryrd} 
\usepackage{subcaption}
\usepackage{thmtools}
\usepackage{tikz-cd} 

\newtheorem{definition}{Definition}

\newtheorem*{principle*}{Principle}
\newtheorem{theorem}{Theorem}
\newtheorem{corollary}{Corollary}

\newtheorem{example}{Example}

\newcommand{\R}{\mathbb{R}}

\usetikzlibrary{arrows}
\usetikzlibrary{calc}
\usetikzlibrary{shapes}
\usetikzlibrary{decorations.pathreplacing}

\begin{document}

\title{On the Bayesian Solution of Differential Equations}
\author{Junyang Wang $^{1}$, Jon Cockayne $^{2}$ and Chris. J. Oates $^{1,3}$ \\
$^{1}$School of Mathematics, Statistics and Physics, Newcastle University, UK \\
$^{2}$Department of Statistics, University of Warwick, UK \\
$^{3}$The Alan Turing Institute, UK}
\maketitle

\abstract{The interpretation of numerical methods, such as finite difference methods for differential equations, as point estimators allows for formal statistical quantification of the error due to discretisation in the numerical context. 
Competing statistical paradigms can be considered and Bayesian probabilistic numerical methods (PNMs) are obtained when Bayesian statistical principles are deployed.
Bayesian PNM are closed under composition, such that uncertainty due to different sources of discretisation can be jointly modelled and rigorously propagated.
However, we argue that no strictly Bayesian PNM for the numerical solution of ordinary differential equations (ODEs) have yet been developed.
To address this gap, we work at a foundational level, where a novel Bayesian PNM is proposed as a proof-of-concept.
Our proposal is a synthesis of classical Lie group methods, to exploit the underlying structure of the gradient field, and non-parametric regression in a transformed solution space for the ODE.
The procedure is presented in detail for first order ODEs and relies on a certain technical condition -- existence of a solvable Lie algebra -- being satisfied.
Numerical illustrations are provided.}

\section{Introduction}

Numerical methods underpin almost all of scientific, engineering and industrial output.
In the abstract, a numerical task can be formulated as the approximation of a quantity of interest 
\begin{eqnarray*}
Q : \mathcal{Y} & \rightarrow & \mathcal{Q}, 
\end{eqnarray*}
subject to a finite computational budget. 
The true underlying state $\mathrm{y}^\dagger \in \mathcal{Y}$ is typically high- or infinite-dimensional, so that only limited information 
\begin{eqnarray*}
A : \mathcal{Y} & \rightarrow & \mathcal{A}
\end{eqnarray*}
on the true state can be exploited for computation and exact computation of $Q(\mathrm{y}^\dagger)$ is prohibited.
For example, $Q(\mathrm{y})$ might be an integral $\int_0^1 \mathrm{y}(x) \mathrm{d}x$ and $A(\mathrm{y})$ might consist of a finite number of evaluations $[\mathrm{y}(x_1) , \dots , \mathrm{y}(x_n)]$ of the integrand.
In this abstract framework a numerical method (e.g. the trapezoidal rule) corresponds to a map $b : \mathcal{A} \rightarrow \mathcal{Q}$, as depicted in Figure \ref{fig: diag1}, where $b(a)$ represents an approximation to the quantity of interest based on the information $a \in \mathcal{A}$.

The increasing ambition and complexity of contemporary applications is such that the computational budget can be \emph{extremely} small compared to the precision that is required at the level of the quantity of interest.
As such, in many important problems it is not possible to reduce the numerical error to a negligible level.
Fields acutely exposed to this challenge include climate forecasting \citep{Wedi2014}, computational cardiology \citep{Chabiniok2016} and molecular dynamics \citep{Perilla2015}.
In the presence of non-negligible numerical error, it is unclear how scientific interpretation of the output of computation can proceed.

\subsection{Probabilistic Numerical Methods} \label{subsec: PNM}

The field of \emph{probabilistic numerics} dates back to \cite{Larkin1972} and aims to provide formal statistical foundations for numerical methods.
Indeed, under the abstract framework just described, numerical methods can be interpreted as point estimators in a statistical context, where a state $\mathrm{y} \in \mathcal{Y}$ can be thought of as a latent variable in a statistical model, and the `data' consist of information $A(\mathrm{y})$ that does not fully determine the quantity of interest $Q(\mathrm{y})$ but is indirectly related to it. 
\cite{Hennig2015} provide an accessible introduction and survey of the field. 

Let the notation $\Sigma_{\mathcal{Y}}$ denote a $\sigma$-algebra on the space $\mathcal{Y}$, assumed fixed, and let $\mathcal{P}_{\mathcal{Y}}$ denote the set of probability measures on $(\mathcal{Y},\Sigma_{\mathcal{Y}})$.
A probabilistic numerical method (PNM) is a procedure which takes as input a `belief' distribution $\mu \in \mathcal{P}_{\mathcal{Y}}$, representing epistemic uncertainty with respect to the true (but unknown) value $\mathrm{y}^\dagger$, along with a finite amount of information, $A(\mathrm{y}^\dagger) \in \mathcal{A}$.
The output is a distribution $B(\mu,A(\mathrm{y}^\dagger)) \in \mathcal{P}_{\mathcal{Q}}$ on $(\mathcal{Q},\Sigma_{\mathcal{Q}})$, representing epistemic uncertainty with respect to the quantity of interest $Q(\mathrm{y}^\dagger)$ after the information $A(\mathrm{y}^\dagger)$ has been processed. 
For example, a PNM for an ordinary differential equation (ODE) takes an initial belief distribution defined on the solution space of the differential equation, together with information arising from a finite number of evaluations of the gradient field, plus the initial condition of the ODE, to produce a distribution over either the solution space of the ODE, or perhaps some derived quantity of interest. 
In this paper, the measurability of $A$ and $Q$ will be assumed.

Despite computational advances in this emergent field, until recently there had not been an attempt to establish rigorous statistical foundations for PNM. 
In \cite{Cockayne2017} the authors argued that Bayesian principles can be adopted.
In brief, this framework requires that the output of a PNM agrees with a rational Bayesian agent:
Let $Q_{\#} : \mathcal{P}_{\mathcal{Y}} \rightarrow \mathcal{P}_{\mathcal{Q}}$ denote the push-forward map associated to $Q$.
i.e. $Q_{\#}(\mu)(S) = \mu(Q^{-1}(S))$ for all $S \in \Sigma_{\mathcal{Q}}$.
Let $\{\mu^a\}_{a \in \mathcal{A}} \subset \mathcal{P}_{\mathcal{Y}}$ denote the disintegration, assumed to exist\footnote{The reader unfamiliar with the concept of a disintegration can safely interpret $\mu^a$ as a technical notion of the `conditional distribution of $\mathrm{y}$ given $A(\mathrm{y}) = a$' in the context of this work. The \emph{disintegration theorem}, Thm. 1 of \cite{Chang1997}, guarantees existence and uniqueness of the disintegration up to a $A_{\#}\mu$-null set under the weak requirements that $\mathcal{Y}$ is a metric space, $\Sigma_{\mathcal{Y}}$ is the Borel $\sigma$-algebra, $\mu$ is Radon, $\Sigma_{\mathcal{A}}$ is countable generated and $\Sigma_{\mathcal{A}}$ contains all singletons $\{a\}$ for $a \in \mathcal{A}$. }, of $\mu \in \mathcal{P}_{\mathcal{Y}}$ along the map $A$.

\begin{figure}[t!]
\centering
\begin{subfigure}{0.3\textwidth}
\centering
\begin{tikzcd}
\mathcal{Y} \arrow[d,"Q",swap] \arrow[rr,"A"] & & \mathcal{A} \arrow[lld,"b"] \\
\mathcal{Q}
\end{tikzcd}
\caption{}
\label{fig: diag1}
\end{subfigure}
\begin{subfigure}{0.3\textwidth}
\centering
\begin{tikzcd}
\mathcal{P}_{\mathcal{Y}} \arrow[d,"Q_{\#}",swap] & & \mathcal{A} \arrow[ll,"\mu^a \mapsfrom a",swap] \arrow[lld,"{B(\mu,a) \mapsfrom a}"] \\
\mathcal{P}_{\mathcal{Q}}
\end{tikzcd}
\caption{}
\label{fig: diag2}
\end{subfigure}
\caption{Diagrams for a numerical method.
(a) The traditional viewpoint of a numerical method is equivalent to a map $b$ from a finite-dimensional information space $\mathcal{A}$ to the space of the quantity of interest $\mathcal{Q}$.
(b) The probabilistic viewpoint treats approximation of $Q(\mathrm{y}^\dagger)$ in a statistical context, described by a map $B(\mu,\cdot)$ from $\mathcal{A}$ to the space of probability distributions on $\mathcal{Q}$.
The probabilistic numerical method $(A,B)$ is Bayesian if and only if (b) is a commutative diagram.}
\end{figure}
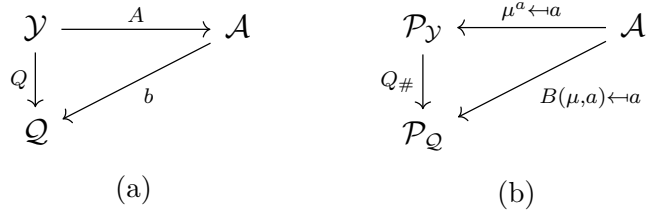

\begin{definition} \label{def: Bayesian}
A probabilistic numerical method $(A,B)$ with $A : \mathcal{Y} \rightarrow \mathcal{A}$ and $B : \mathcal{P}_{\mathcal{Y}} \times \mathcal{A} \rightarrow \mathcal{P}_{\mathcal{Q}}$ for a quantity of interest $Q : \mathcal{Y} \rightarrow \mathcal{Q}$ is \emph{Bayesian} if and only if $B(\mu,a) = Q_{\#}(\mu^a)$ for all $\mu \in \mathcal{P}_{\mathcal{Y}}$ and all $a \in \mathcal{A}$.
\end{definition}

This definition is intuitive; the output of the PNM should coincide with the marginal distribution for $Q(\mathrm{y}^\dagger)$ according to the disintegration element $\mu^a \in \mathcal{P}_{\mathcal{Y}}$, based on the information $a \in \mathcal{A}$ that was provided. 
The definition is equivalent to the statement that Figure \ref{fig: diag2} is a commutative diagram.
In \cite{Cockayne2017} the map $A$ was termed an \emph{information operator} and the map $B$ was termed a \emph{belief update operator}; we adhere to these definitions in our work.

The Bayesian approach to PNM confers several important benefits:
\begin{itemize}
\item The input $\mu$ and output $B(\mu,a)$ belief distributions can be interpreted, respectively, as a \emph{prior} and (marginal) \emph{posterior}. 
As such, they automatically inherit the stronger formal semantics and philosophical foundations that underpin the Bayesian framework and, in this sense, are well-understood \citep[see e.g.][]{Gelman2013}.
\item The definition of Bayesian PNM is operational.
Thus, if we are presented with a prior $\mu$ and information $a$ then there is a unique Bayesian PNM and it is constructively defined.
\item The class of Bayesian PNM is closed under composition, such that uncertainty due to different sources of discretisation can be jointly modelled and rigorously propagated.
This point will not be discussed further in this work, but we refer the interested reader to Section 5 of \cite{Cockayne2017}.
\end{itemize}

Nevertheless, the strict definition of Bayesian PNM limits scope to design convenient computational algorithms and indeed several proposed PNM are not Bayesian \citep[see Table 1 in][]{Cockayne2017}.
The challenge is two-fold; for a Bayesian PNM, the elicitation of an appropriate prior distribution $\mu$ and the exact computation of its disintegration $\{\mu^a\}_{a \in \mathcal{A}}$ must both be addressed. 
In the next section we argue that -- perhaps as a consequence of these constraints -- a strictly Bayesian PNM for the numerical solution of an ODE does not yet exist.

\subsection{Early Work on PNM for ODEs} \label{subsec: existing work}

Consider a generic univariate first-order initial value problem
\begin{eqnarray}
\frac{\mathrm{d}\mathrm{y}}{\mathrm{d}x} & = & f(x,\mathrm{y}(x)), \qquad x \in [x_0,x_T], \qquad \mathrm{y}(x_0) = y_0 \label{eq: ODE1} 
\end{eqnarray}
The notational convention used in this paper is that $\mathrm{y}$ denotes a generic function, whereas the italicised $y$ denotes a generic value taken by that function.
Throughout this paper all ODEs that we consider will be assumed to be well-defined and admit a unique solution $\mathrm{y}^\dagger \in \mathcal{Y}$ where $\mathcal{Y}$ is some pre-specified set.
In this paper the quantity of interest $Q(\mathrm{y}^\dagger)$ will be the solution curve $\mathrm{y}^\dagger$ itself.

The first probabilistic perspective on ODEs, of which we are aware, was \cite{Skilling1992}.
Originally described as `Bayesian' by the author, we will argue that, at least in the strict sense of Definition \ref{def: Bayesian}, it is not a Bayesian PNM.
In general, a PNM attempts to reason about $\mathrm{y}^\dagger$ based on evaluations $f(x_i,y_i)$ of the gradient field for certain input pairs $\{(x_i,y_i)\}_{i=1}^n$.
Let $a_i = f(x_i,y_i)$ and $a^i = [a_0,\dots,a_i]$.
The selection of the input pairs $(x_i,y_i)$ on which $f$ is evaluated is not constrained and several possibilities, of increasing complexity, were discussed in \cite{Skilling1992}.
To fix ideas, the simplest such approach is to proceed iteratively as follows:
\begin{enumerate}
\item[(0.1)] The first pair $(x_0,y_0)$ is fully determined by the initial condition of the ODE.
\item[(0.2)] This provides one piece of information, $a_0 = f(x_0,y_0)$.
\item[(0.3)] The belief $\mu$ is updated according to $a_0$, leading to a new belief $\mu_0$ which is just the disintegration element $\mu^{a^0}$.
\item[(1)] A discrete time step $x_1 = x_0 + h$, where $h = \frac{x_T - x_0}{n} > 0$, is performed and a particular point estimate $y_1 = \int \mathrm{y}(x_1) \mathrm{d} \mu_0(\mathrm{y})$ for the unknown true value $\mathrm{y}^\dagger(x_1)$ is obtained.
This specifies the second pair $(x_1,y_1)$.
(Here $y_1$ is the predictive mean for $\mathrm{y}(x_1)$ based on the current belief $\mu_0$.)
\end{enumerate}
The process continues similarly, such that at time step $i-1$ we have a belief distribution $\mu_{i-1} = B(\mu,a^{i-1}) \in \mathcal{P}_{\mathcal{Y}}$, where the general belief update operator $B$ is yet to be defined, and the following step is performed:
\begin{enumerate}
\item[($i$)] Let $x_i = x_{i-1} + h$ and set $y_i = \int \mathrm{y}(x_i) \mathrm{d}\mu_{i-1}(\mathrm{y})$ .
\end{enumerate}
The final output is a probability distribution $\mu_n = B(\mu,a^n) \in \mathcal{P}_{\mathcal{Y}}$.

The method just described is \emph{not} actually a PNM in the concrete sense that we have defined.
Indeed, the final output $\mu_n$ is a deterministic function of the values $a^n$ of the gradient field that were obtained.
However, the values of the gradient field $f(x,y)$ outside any open neighbourhood of the true solution curve $\mathcal{C} = \{(x,y) : y = \mathrm{y}^\dagger(x), \; x \in [x_0,x_T]\}$ do not determine the solution of the ODE and, conversely, the solution of the ODE provides no information about the values of the gradient field outside any open neighbourhood of the true solution curve $\mathcal{C}$.
Thus it is not possible, in general, to write down an information operator $A : \mathcal{Y} \rightarrow \mathcal{A}$ that reproduces the information $a^n$ when applied to the solution curve $\mathrm{y}^\dagger(\cdot)$ of the ODE.

The approach taken in \cite{Skilling1992}, cast in our abstract framework, was therefore to posit an \emph{approximate} information operator 
$$\hat{A}(\mathrm{y}) = \left[\frac{\mathrm{d}\mathrm{y}}{\mathrm{d}x}(x_0), \\ \dots \\ ,\frac{\mathrm{d}\mathrm{y}}{\mathrm{d}x}(x_n)\right]$$
Of course, $\hat{A}(\mathrm{y}^\dagger) \neq a^n$ in general.
To acknowledge the approximation error, \cite{Skilling1992} proposed a particular belief update operator $B$ that attempts to model the information with a Gaussian potential:
\begin{eqnarray}
\frac{\mathrm{d}\mu_n}{\mathrm{d}\mu_0}(\mathrm{y}) = \prod_{i=1}^n\frac{\mathrm{d}\mu_i}{\mathrm{d}\mu_{i-1}}(\mathrm{y}), \qquad \frac{\mathrm{d}\mu_i}{\mathrm{d}\mu_{i-1}}(\mathrm{y}) \propto \exp\left( -\frac{1}{2 \sigma^2} \left( \frac{\mathrm{d}\mathrm{y}}{\mathrm{d}x}(x_i) - f(x_i,y_i) \right)^2 \right)  \label{eq: Skilling Lhood}
\end{eqnarray}
This Gaussian potential was referred to in \cite{Skilling1992} as a `likelihood' and, together with $\mu_0 = \mu^{a^0}$, the output $\mu_n$ is completely specified.
Here $\sigma$ is a fixed positive constant, however in principle a non-diagonal covariance matrix can also be considered.

The negative consequences of basing inferences on an approximate information operator $\hat{A}$ are potentially twofold.
First, in the special case where the gradient field $f$ does not depend the second argument, the quantities $\frac{\mathrm{d}\mathrm{y}}{\mathrm{d}x}(x_i)$ and $f(x_i,y_i)$ are identical.
From this perspective, $\mu_n$ represents inference under a mis-specified likelihood, since information is treated as erroneous when it is in fact exact.
Second, recall that values of the gradient field that are not contained on the true solution curve $\mathcal{C}$ of the ODE are \emph{ancillary} with respect to the quantity of interest $\mathrm{y}^\dagger$.
It follows that dependence of the output $\mu_n$ on these ancillary values violates the \emph{conditionality principle} (CP) \citep[][Sec. 2.2]{Cox1974}.
Since the CP is an implication of the \emph{likelihood principle}, violation of the CP implies violation of the Bayesian framework.
This confirms, through a different argument, that the approach of \cite{Skilling1992} cannot be Bayesian in the sense of Definition \ref{def: Bayesian}.

\subsection{Recent Work on PNM for ODEs} 

Inspired by \cite{Skilling1992}, several variations on the method have been proposed:
\begin{itemize}
\item The approach of \cite{Schober2014} considered Eq.~\eqref{eq: Skilling Lhood} in the $\sigma \downarrow 0$ limit.
The authors proved that if the input pair $(x_1,y_1)$ is  taken as $y_1 = \int \mathrm{y}(x_1) \mathrm{d}\mu_0(\mathrm{y})$, and a particular belief $\mu$ was used, then the smoothing estimate $\hat{y}_1 = \int \mathrm{y}(x_1) \mathrm{d} \mu_1(\mathrm{y})$, i.e. the posterior mean for $\mathrm{y}(x_1)$ based on information $a^1$, coincides with the deterministic approximation to $\mathrm{y}^\dagger(x_1)$ that would be provided by a $k$-th order Runge-Kutta method.
Similar connections to multistep methods of Nordsieck and Adams form were identified, respectively, in \cite{Schober2016} and \cite{Teymur2016,Teymur2018}.
\item The work of \cite{Kersting2016} proposed an adaptive choice of covariance matrix for use in Eq.~\eqref{eq: Skilling Lhood}, in order to encourage uncertainty estimates to be better calibrated. 
\item The original work of \cite{Chkrebtii2013} is somewhat related to \cite{Kersting2016}, however instead of using the mean of the current posterior as input to the gradient field, the input pair $(x_i,y_i)$ was selected by sampling $y_i$ from the marginal distribution for $\mathrm{y}(x_i)$ implied by $\mu_{i-1}$. 
\item The approaches proposed in \cite{Conrad2015,Abdulle2018} are not motivated in the Bayesian framework, but instead seek to introduce a stochastic perturbation into a classical numerical method.
\end{itemize}
The above methods provide elegant and practical uncertainty quantification for ODEs, but do not fulfil the requirements of a Bayesian PNM.

 \subsection{Our Contribution}

This paper presents a proof-of-concept PNM for the numerical solution of ODEs that is strictly Bayesian in the sense of Definition \ref{def: Bayesian}.
For an $n$th order ODE, we require the ODE to admit an $n$-dimensional solvable Lie algebra.
Then, our proposal is to perform exact Bayesian inference on a particular Lie-transformed ODE, whose gradient field is a function of the independent state variable only, in effect reducing the ODE to an integral (see Fig. \ref{fig: schematic}).
Only the case of first order ODE is presented in this paper, due to page limitations, but the method itself is general\footnote{Full details for second- and higher-order ODEs will be presented in a subsequent extended manuscript}.
In addition to the benefits conferred in the Bayesian framework, detailed in Section \ref{subsec: PNM} and in \cite{Cockayne2017}, the method being proposed can be computationally realised, but at a substantially increased run-time compared to existing, non-Bayesian approaches.
As such, we consider this work to be a proof-of-concept rather than an applicable Bayesian PNM.

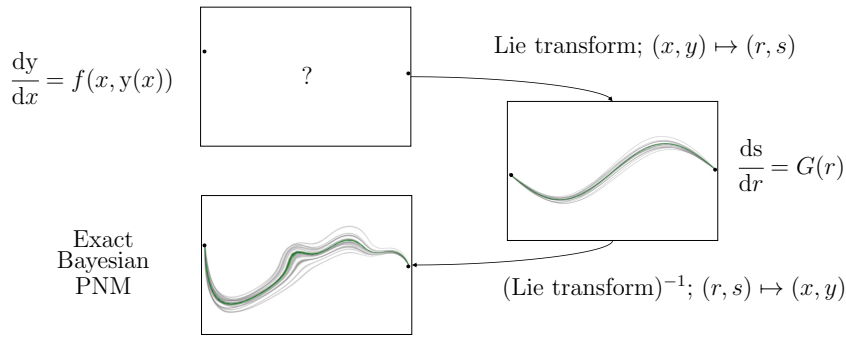
\begin{figure}[t!]
\centering
\resizebox{.7\textwidth}{!}{
\begin{tikzpicture}
\node[anchor=south west,inner sep=0] (image) at (0,0) {
\includegraphics[width = 0.6\textwidth,clip,trim = 4cm 2cm 4cm 2cm,page=1]{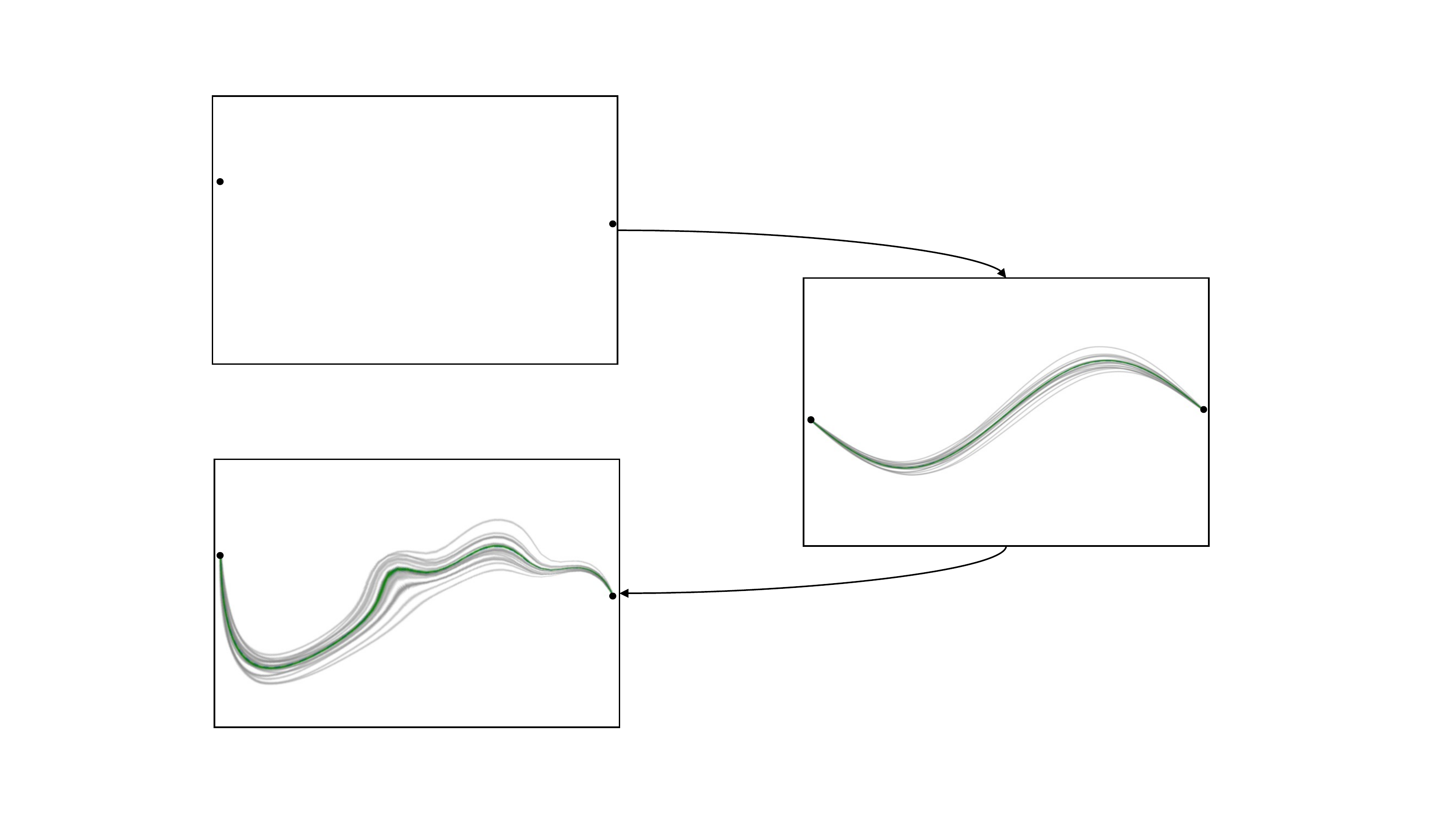}};
\node (A) at (-1.5,4.5) {$\cfrac{\mathrm{d}\mathrm{y}}{\mathrm{d}x} = f(x,\mathrm{y}(x))$};
\node (B) at (10.5,3) {$\cfrac{\mathrm{d}\mathrm{s}}{\mathrm{d}r} = G(r)$};
\node (C) at (8,5) {Lie transform; $(x,y) \mapsto (r,s)$};
\node (D) at (8.5,0.85) {(Lie transform)$^{-1}$; $(r,s) \mapsto (x,y)$};
\node (E) at (-1.3,1.7) {Exact};
\node (F) at (-1.3,1.3) {Bayesian};
\node (G) at (-1.3,0.9) {PNM};
\node (H) at (2.2,4.5) {?};
\end{tikzpicture} }
\caption{Schematic of our proposed approach.
An $n$th order ODE that admits a solvable Lie algebra can be transformed into $n$ integrals, to which exact Bayesian probabilistic numerical methods can be applied.
The posterior measure on the transformed space is then pushed back through the inverse transformation onto the original domain of interest.
}
\label{fig: schematic}
\end{figure}
 
\section{Methods} \label{sec: methods}

In this section our novel Bayesian PNM is presented for the case of First Order ODEs. This allows some of the more technical details associated to the general case to be omitted, due to the fact that any one-dimensional Lie algebra is trivially solvable. 
Due to page restrictions, an overview of basic Lie group methods is reserved for the Appendix, and we also refer the reader to chapters 1-3 of \cite{Bluman2002} for complete detail.

The main result from Lie group methods that will allow us to construct an exact Bayesian PNM is as follows:

\begin{theorem}[Reduction of a First Order ODE to an Integral]\label{1stordertoquad}
If a first order ODE 
\begin{eqnarray}\label{1storder}
\frac{\mathrm{d}\mathrm{y}}{\mathrm{d}x} & = & f(x,\mathrm{y}(x)), \qquad \mathrm{y} \in \mathcal{Y}
\end{eqnarray}
admits a one parameter Lie group of transformations, then there exists canonical coordinates $r(x,y)$, $s(x,y)$ such that 
\begin{eqnarray}
\frac{\mathrm{d}\mathrm{s}}{\mathrm{d}r} & = & G(r), \qquad \mathrm{s} \in \mathcal{S} \label{eq: G integral}
\end{eqnarray}
for some explicit function $G(r)$.
\end{theorem}
\begin{proof}
See Appendix.
\end{proof}

Note that the transformed ODE in Eq.~\eqref{eq: G integral} is nothing more than an integral, for which exact Bayesian PNM can be used \citep[e.g.][]{Briol2016}.
In general, if the Lie algebra exhibited by the symmetries of an $n$th order ODE contains a $n$-dimensional solvable subalgebra, then that ODE can be reduced to $n$ nested integrals and the method we are about to present can be employed.
At a high level, as indicated in Fig. \ref{fig: schematic}, our proposed PNM performs exact Bayesian inference for the solution $\mathrm{s} \in \mathcal{S}$ of Eq.~\eqref{eq: G integral} and then transforms the resultant posterior in $\mathcal{P}_{\mathcal{S}}$ back into the original coordinate system to obtain an element in $\mathcal{P}_{\mathcal{Y}}$.
From Def. \ref{def: Bayesian}, our Bayesian PNM is uniquely determined by the combination of the prior $\nu \in \mathcal{P}_{\mathcal{S}}$ and the information operator
\begin{equation} \label{eq: info operator 1st order}
A(\mathrm{y}) = [G(r_0),...,G(r_n)] \in \mathcal{A} = \mathbb{R}^{n+1} 
\end{equation}
which corresponds indirectly to $n+1$ evaluations of the original gradient field $f$ at certain input pairs $(x_i,y_i)$.
The selection of the inputs $r_i$ is in principle unrestricted.
The transformation of a first order ODE is now briefly illustrated:
\begin{example}\label{1storderexample}
Consider the first order ODE
\numberwithin{equation}{example}
\begin{equation}\label{eq:1storderode}
\frac{\mathrm{d}\mathrm{y}}{\mathrm{d}x} \; = \; f(x,\mathrm{y}(x)), \qquad f(x,y) \; = \; F\left( \frac{y}{x} \right) .
\end{equation}
This ODE has the canonical coordinates $s = \log y$, $r = \frac{y}{x}$.
The transformed ODE is then
\begin{eqnarray}\label{eq:1storderodecanon}
\frac{\mathrm{d}\mathrm{s}}{\mathrm{d}r} & = & \cfrac{F(r)}{-r^2+rF(r)} \; =: \; G(r).
\end{eqnarray}
Thus an evaluation $G(r)$ corresponds to an evaluation of $f(x,y)$ at an input $(x,y)$ such that $r = \frac{y}{x}$.
\end{example}

The approach just described cannot proceed unless the prior $\nu \in \mathcal{P}_{\mathcal{S}}$ corresponds, implicitly, to a well-defined distribution $\mu \in \mathcal{P}_{\mathcal{Y}}$ in the original coordinate system $\mathcal{Y}$.
This precludes standard (e.g. Gaussian process) priors in general, as such priors assign mass to functions in $(r,s)$-space that do not correspond to well-defined functions in $(x,y)$-space (see Fig. \ref{fig: not well defined}).

\begin{figure}[t!]
\centering
\resizebox{.6\textwidth}{!}{
\begin{tikzpicture}
\node[anchor=south west,inner sep=0] (image) at (0,0) {
\includegraphics[width = 0.6\textwidth,clip,trim = 3.6cm 4cm 4cm 4cm,page=2]{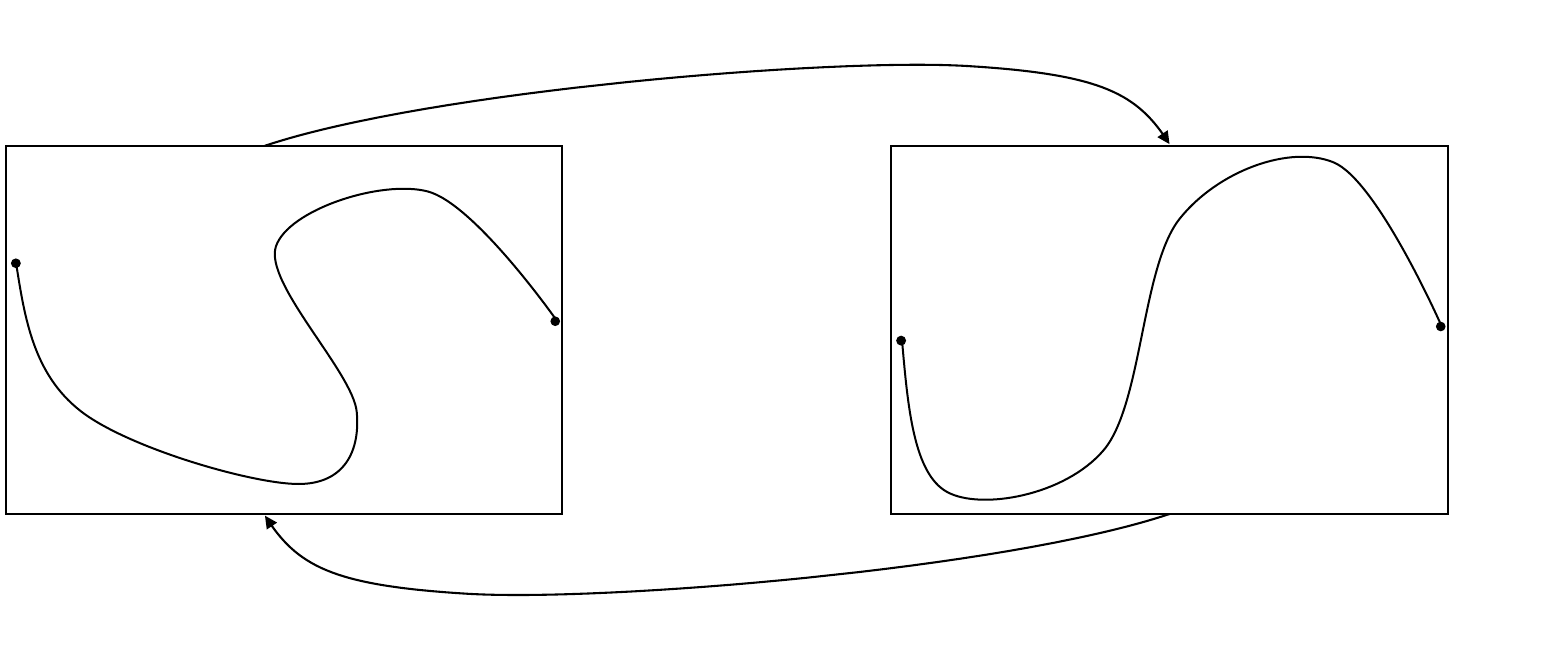}};
\node (A) at (-1.3,2.1) {$(x,y)$};
\node (B) at (10.5,2.1) {$(r,s)$};
\node (C) at (4.5,4.2) {Lie transform};
\node (D) at (4.5,0.05) {(Lie transform)$^{-1}$};
\end{tikzpicture} }
\caption{Illustration of the implicit prior principle: A prior elicited for the function $\mathrm{s}(r)$ in the transformed coordinate system $(r,s)$ must be supported on functions $\mathrm{s}(r)$ that correspond to well-defined functions $\mathrm{y}(x)$ in the original coordinate system $(x,y)$. Thus the situation depicted would not be allowed.}
\label{fig: not well defined}
\end{figure}
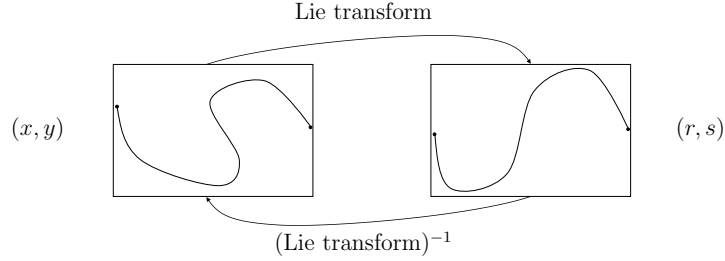

For such an implicit prior to be well-defined we need to understand when a function in $(r,s)$ space maps to a well-defined function in the original $(x,y)$ domain of interest.

\begin{principle*}[Implicit Prior]
A distribution $\nu \in \mathcal{P}_{\mathcal{S}}$ on the transformed solution space $\mathcal{S}$ corresponds to a well-defined \emph{implicit prior} $\mu \in \mathcal{P}_{\mathcal{Y}}$ provided that $x(r,\mathrm{s}(r))$ is strictly monotone as a function of $r$, or equivalently (without loss of generality) $\mathrm{d}x / \mathrm{d}r > 0$ for all $r$.
\end{principle*}

\begin{example}[Ex. \ref{1storderexample}, continued] \label{1storderexample again 2}
For the ODE in Ex. \ref{1storderexample}, with canonical coordinates $s = \log y$, $r = \frac{y}{x}$, if $x \in [x_0,x_T] = [1,x_T]$ and $y \in (0,\infty)$, then the region in the $(r,s)$ plane corresponding to $[1,x_T] \times (0,\infty)$ in the $(x,y)$ plane is $(0,\infty) \times \R$.
Now, 
\begin{eqnarray*}
\frac{\mathrm{d}x(r,\mathrm{s}(r))}{\mathrm{d}r} & = & \frac{r \mathrm{s}'(r) \exp(\mathrm{s}(r))-\exp(\mathrm{s}(r))}{r^2} .
\end{eqnarray*}
Thus $\mathrm{d}x / \mathrm{d}r > 0 \Leftrightarrow \mathrm{s}'(r) \; > \; 1 / r$ and the implicit prior principle requires that we respect the constraint
\begin{equation}\label{eq:domaincond}
\log(r) \le \mathrm{s}(r) \le \log(r)+\log(x_T) \quad \forall r>0.
\end{equation}
The prior $\nu \in \mathcal{P}_{\mathcal{S}}$ must therefore be supported on differentiable functions $\mathrm{s}$ on $r \in (0,\infty)$ and satisfying Eq.~\eqref{eq:domaincond}.
\end{example}

\section{Experimental Results} \label{sec: experiment}

In this section the proposed Bayesian PNM is experimentally tested.
Again the case of a first order ODE is considered. Scope is limited to verifying the correctness of the procedure, as well as indicating how implicit prior distributions can be constructed.

To limit scope, we consider our running example from Sec. \ref{sec: methods}:
\begin{equation}
\frac{\mathrm{d}\mathrm{y}}{\mathrm{d}x} = F\left( \frac{\mathrm{y}(x)}{x} \right) , \qquad x \in [1,x_T], \qquad \mathrm{y}(1) = y_0 . \label{eq: experiment 1st order}
\end{equation}
Note the associated canonical coordinates for this class of ODE have already been derived in Ex. \ref{1storderexample}.
In constructing a prior $\mu \in \mathcal{P}_{\mathcal{Y}}$, we must respect the implicit prior principle in Sec. \ref{sec: methods}.
Indeed, recall from Ex. \ref{1storderexample} that the ODE in Eq.~\eqref{eq: experiment 1st order} can be transformed into an ODE of the form
\begin{eqnarray*}
\frac{\mathrm{d} \mathrm{s}}{\mathrm{d} r} = G(r), \qquad r \in (0,\infty) , \qquad \mathrm{s}(y_0) = \log(y_0) .
\end{eqnarray*}
Then our approach constructs a distribution $\nu \in \mathcal{P}_{\mathcal{S}}$ where, from Ex. \ref{1storderexample again 2}, $\mathcal{S}$ is the set of differentiable functions $\mathrm{s}$ defined on $r \in (0,\infty)$ and satisfying Eq.~\eqref{eq:domaincond}.
Note that the constraints in Eq.~\eqref{eq:domaincond} preclude the direct use of standard prior models, such as Gaussian processes.
However, it is nevertheless possible to design priors that are convenient for a given set of canonical coordinates.
Indeed, for the canonical coordinates $r,s$ in our example, we can consider a prior of the form $\mathrm{s}(r) = \log(r) + \log(x_T)\zeta(r)$ where the function $\zeta : (0,\infty) \rightarrow \mathbb{R}$ satisfies $\zeta(y_0)=0$, $\zeta(r) \leq 1$ and $\mathrm{d}\zeta / \mathrm{d}r \geq 0$.
For this experiment, the approach of \cite{Lopez-Lopera2017} was used as a prior model for the monotone, bounded function $\zeta$; for brevity we refer the reader to the paper for further detail.

To obtain empirical results we consider the ODE with $F(r) = r^{-1} + r$ and $y_0 = 0.1$.
The function $F$ was evaluated on a regular grid of $n$ points.
The posterior distributions that were obtained as the number $n$ of data points was increased were plotted in the $(r,s)$ plane in Fig. \ref{fig: 1st order results rs} and in the $(x,y)$ plane in Fig. \ref{fig: 1st order results xy}.
Observe that the implicit prior principle ensures that all curves in the $(x,y)$ plane are well-defined functions (i.e. there is at most one $y$ value for each $x$ value).
Observe also that the posterior mass contracts to the true solution of the ODE as the number of evaluations $n$ of the gradient field is increased.

\begin{figure}[t!]
\centering
\includegraphics[width = 0.3\textwidth,clip,trim = 5cm 9.5cm 5cm 9.5cm]{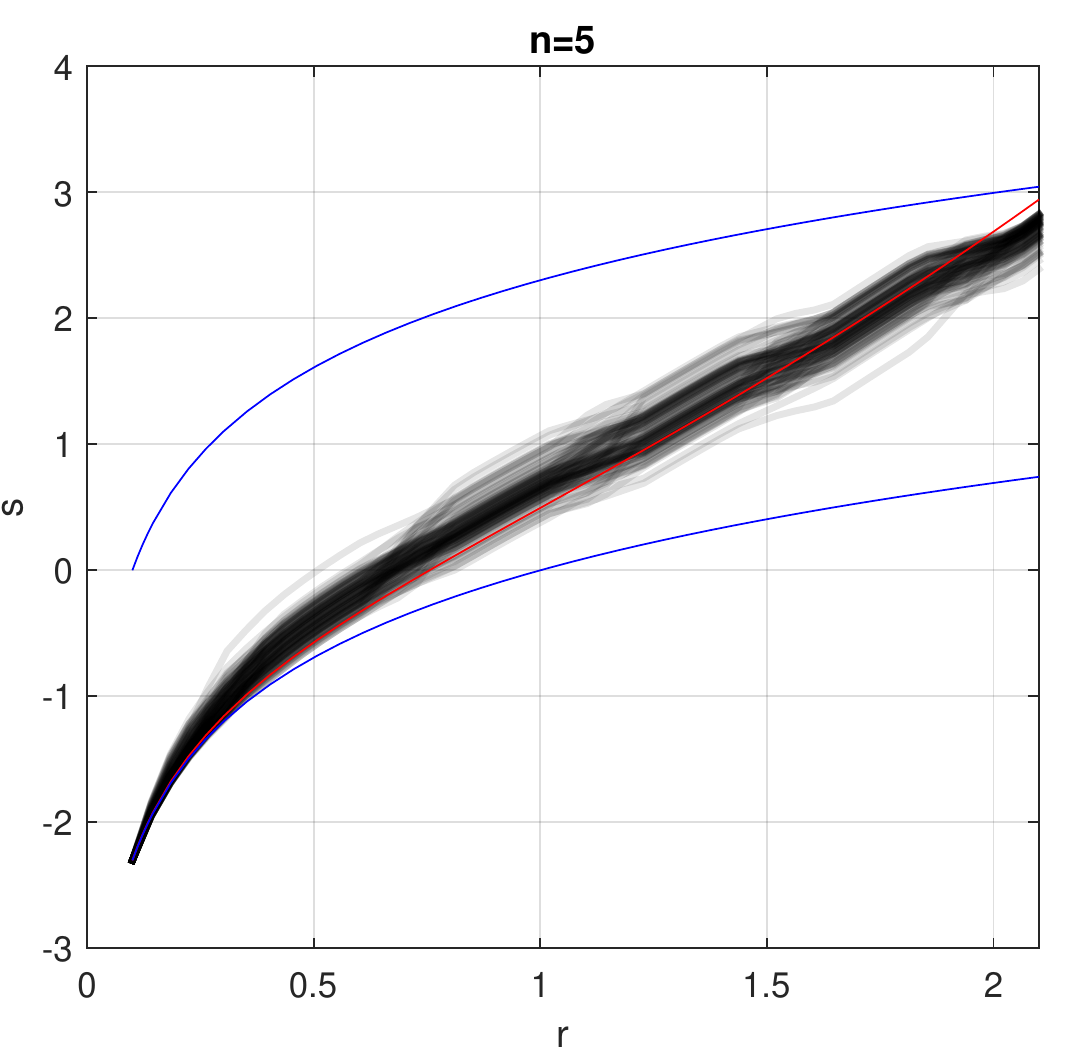} \hspace{30pt}
\includegraphics[width = 0.3\textwidth,clip,trim = 5cm 9.5cm 5cm 9.5cm]{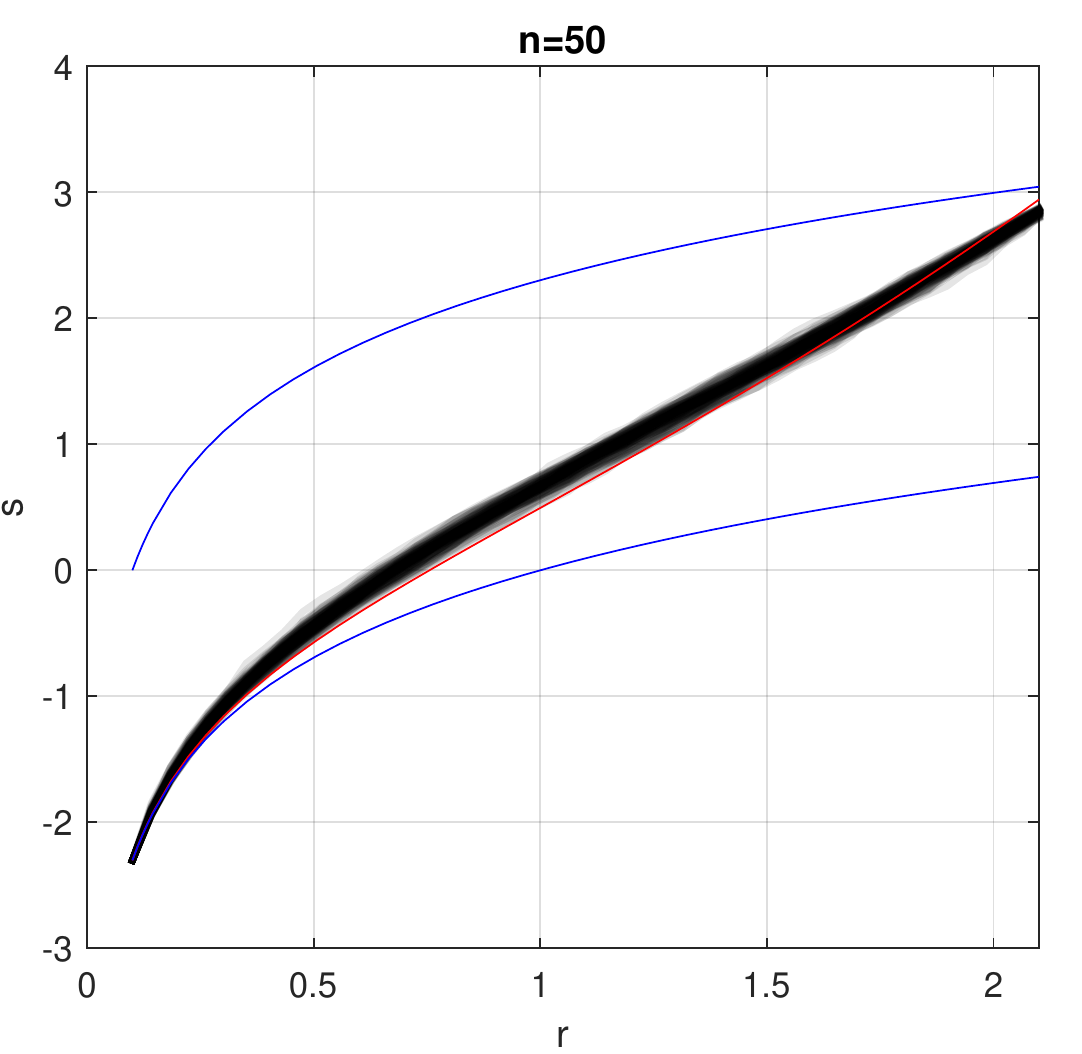}
\caption{Experimental results, first order ODE: 
The black curves represent samples from the posterior in the $(r,s)$ plane, whilst the exact solution is indicated in red.
The blue curves represent a constraint on the domain that arises when the implicit prior principle is applied.
The number $n$ of gradient field evaluations is indicated.}
\label{fig: 1st order results rs}
\end{figure}

\begin{figure}[t!]
\centering
\includegraphics[width = 0.3\textwidth,clip,trim = 5cm 9.5cm 5cm 9.5cm]{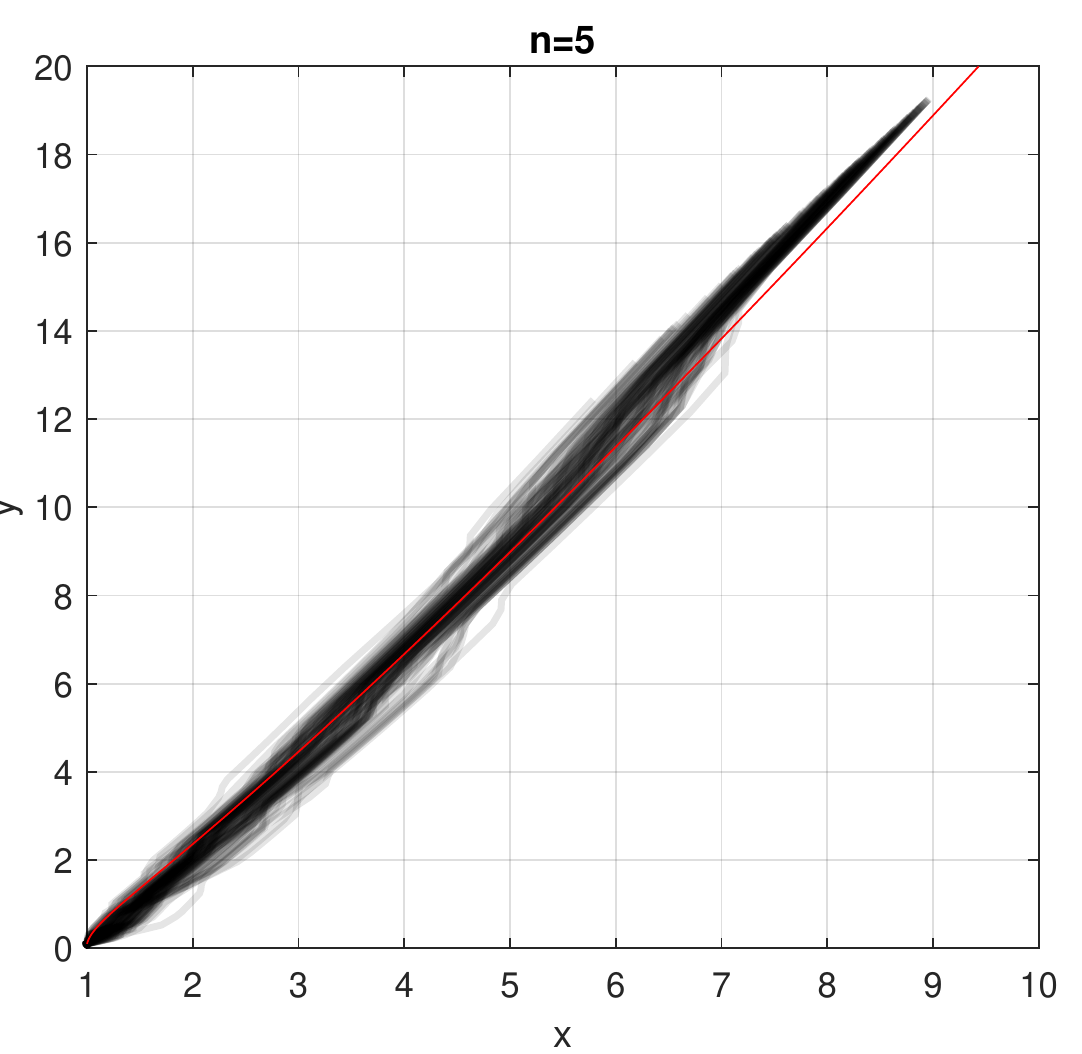} \hspace{30pt}
\includegraphics[width = 0.3\textwidth,clip,trim = 5cm 9.5cm 5cm 9.5cm]{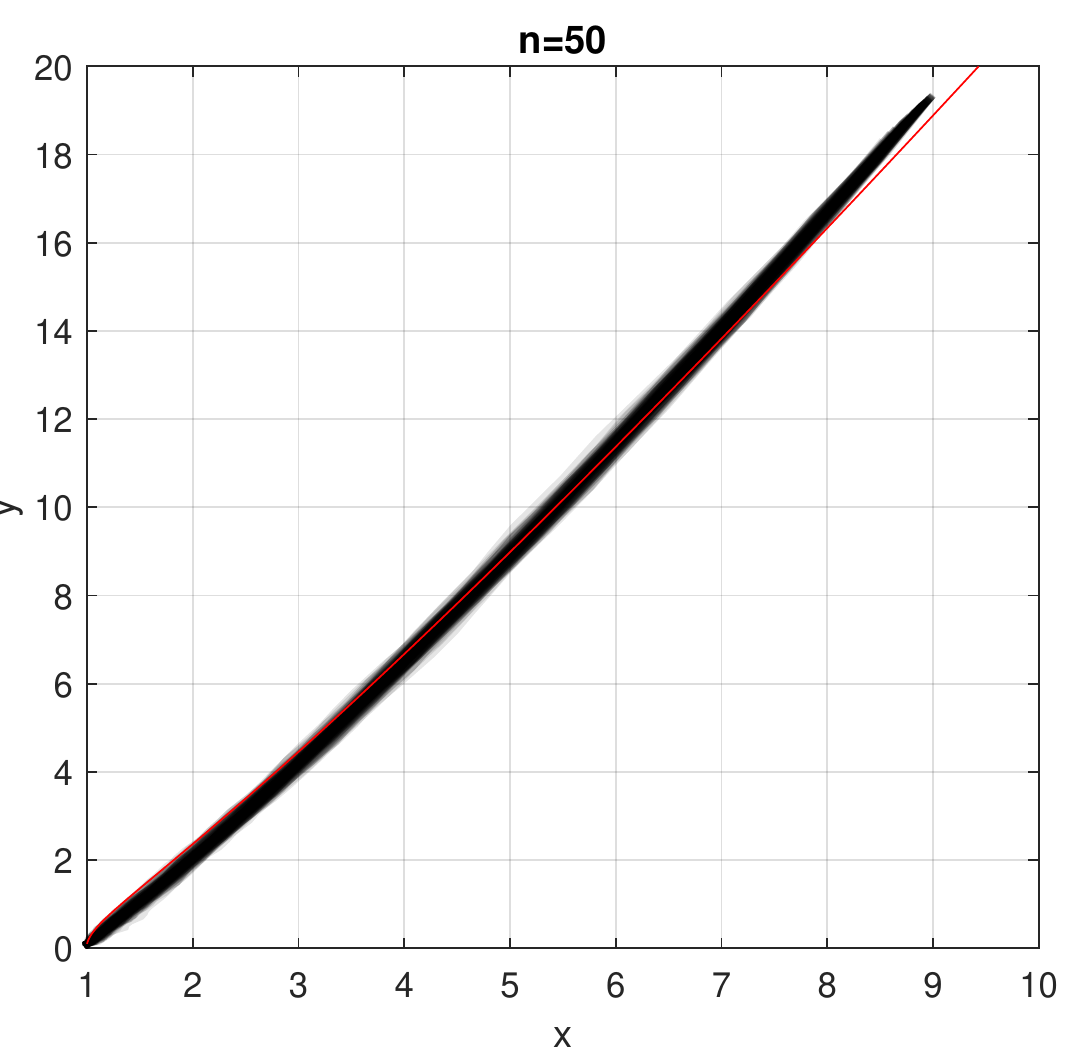}
\caption{Experimental results, first order ODE: 
The black curves represent samples from the posterior in the $(x,y)$ plane, whilst the exact solution is indicated in red.
The number $n$ of gradient field evaluations is indicated.}
\label{fig: 1st order results xy}
\end{figure}

\section{Conclusion} \label{sec: conclusion}

This paper presented a foundational perspective on the solution of ODEs by a PNM.
It was first argued that there did not yet exist a Bayesian PNM in this context.
Then, to address this gap, a novel Bayesian PNM was developed.
The Bayesian perspective that we have put forward sheds light on foundational issues which will need to be addressed going forward:

As explained in Section \ref{subsec: existing work}, existing PNM for ODEs each take the underlying state space $\mathcal{Y}$ to be the solution space of the ODE.
This appears to be problematic, in the sense that a generic evaluation $f(x_i,y_i)$ of the gradient field cannot be cast as information $A(\mathrm{y}^\dagger)$ about the solution $\mathrm{y}^\dagger$ of the ODE unless the point $(x_i,y_i)$ lies exactly on the solution curve $\{(x,\mathrm{y}^\dagger(x)) : x \in [x_0,x_T]\}$.
As a consequence, all existing PNM of which we are aware operate outside the Bayesian framework.
The assumption of a solvable Lie algebra, used in this work, can be seen as a mechanism to ensure the existence of an exact information operator $A$, so that this problem is avoided.

The proposed method was intended as a proof-of-concept and it is therefore useful to highlight the aspects in which it is limited. First, the route to obtain transformations admitted by the ODE demands that some aspects of the gradient field $f$ are known, in contrast to other work in which $f$ is treated as a black-box.
Second, the class of ODEs for which a solvable Lie algebra is admitted is relatively small, though references such as \cite{Bluman2002} document important cases where our method could be applied.
Third, the principles for prior construction that we identified do not entail a unique prior and, as such, the question of prior elicitation must still be addressed. 
Nevertheless, as our proof-of-concept has demonstrated, the development of an exact Bayesian method, or a principled approximation to such a method, is not an insurmountable task.

\vspace{5pt}
\noindent{\bf Acknowledgements:}
The authors are grateful to Mark Craddock, Fran\c{c}ois-Xavier Briol and Tim Sullivan for discussion of this work.
JW was supported by the EPSRC Centre for Doctoral Training in Cloud Computing for Big Data at Newcastle University, UK.
CJO was supported by the Lloyd's Register Foundation programme on data-centric engineering at the Alan Turing Institute, UK.
This material was based upon work partially supported by the National Science Foundation under Grant DMS-1127914 to the Statistical and Applied Mathematical Sciences Institute. Any opinions, findings, and conclusions or recommendations expressed in this material are those of the authors and do not necessarily reflect the views of the National Science Foundation.

\bibliographystyle{plainnat}
\bibliography{bibliography}

\begin{thebibliography}{21}
\providecommand{\natexlab}[1]{#1}
\providecommand{\url}[1]{\texttt{#1}}
\expandafter\ifx\csname urlstyle\endcsname\relax
  \providecommand{\doi}[1]{doi: #1}\else
  \providecommand{\doi}{doi: \begingroup \urlstyle{rm}\Url}\fi

\bibitem[Abdulle and Garegnani(2018)]{Abdulle2018}
Assyr Abdulle and Giacomo Garegnani.
\newblock Random time step probabilistic methods for uncertainty quantification
  in chaotic and geometric numerical integration.
\newblock \emph{arXiv preprint:1801.01340}, 2018.

\bibitem[Bluman and Anco(2002)]{Bluman2002}
G.W. Bluman and S.C. Anco.
\newblock \emph{Symmetry and Integration Methods for Differential Equations}.
\newblock Springer, 2002.

\bibitem[Briol et~al.(2015)Briol, Oates, Girolami, Osborne, and
  Sejdinovic]{Briol2016}
Fran{\c{c}}ois-Xavier Briol, Chris~J Oates, Mark Girolami, Michael~A Osborne,
  and Dino Sejdinovic.
\newblock Probabilistic integration: A role in statistical computation?
\newblock \emph{arXiv preprint:1512.00933}, 2015.

\bibitem[Chabiniok et~al.(2016)Chabiniok, Wang, Hadjicharalambous, Asner, Lee,
  Sermesant, Kuhl, Young, Moireau, Nash, et~al.]{Chabiniok2016}
Radomir Chabiniok, Vicky~Y Wang, Myrianthi Hadjicharalambous, Liya Asner, Jack
  Lee, Maxime Sermesant, Ellen Kuhl, Alistair~A Young, Philippe Moireau,
  Martyn~P Nash, et~al.
\newblock Multiphysics and multiscale modelling, data--model fusion and
  integration of organ physiology in the clinic: ventricular cardiac mechanics.
\newblock \emph{Interface Focus}, 6\penalty0 (2):\penalty0 20150083, 2016.

\bibitem[Chang and Pollard(1997)]{Chang1997}
Joseph~T Chang and David Pollard.
\newblock Conditioning as disintegration.
\newblock \emph{Statistica Neerlandica}, 51\penalty0 (3):\penalty0 287--317,
  1997.

\bibitem[Chkrebtii et~al.(2016)Chkrebtii, Campbell, Girolami, and
  Calderhead]{Chkrebtii2013}
Oksana Chkrebtii, David~A Campbell, Mark~A Girolami, and Ben Calderhead.
\newblock Bayesian uncertainty quantification for differential equations (with
  discussion).
\newblock \emph{Bayesian Analysis}, 11\penalty0 (4):\penalty0 1239--1267, 2016.

\bibitem[Cockayne et~al.(2017)Cockayne, Oates, Sullivan, and
  Girolami]{Cockayne2017}
Jon Cockayne, Chris Oates, Tim Sullivan, and Mark Girolami.
\newblock Bayesian probabilistic numerical methods.
\newblock \emph{arXiv preprint:1702.03673}, 2017.

\bibitem[Conrad et~al.(2017)Conrad, Girolami, S{\"a}rkk{\"a}, Stuart, and
  Zygalakis]{Conrad2015}
Patrick~R Conrad, Mark Girolami, Simo S{\"a}rkk{\"a}, Andrew Stuart, and
  Konstantinos Zygalakis.
\newblock Statistical analysis of differential equations: introducing
  probability measures on numerical solutions.
\newblock \emph{Statistics and Computing}, 27\penalty0 (4):\penalty0
  1065--1082, 2017.

\bibitem[Cox and Hinkley(1974)]{Cox1974}
David~R Cox and David~V Hinkley.
\newblock \emph{Theoretical Statistics}, volume~1.
\newblock London, Chapman \& Hall, 1974.

\bibitem[Gelman and Shalizi(2013)]{Gelman2013}
Andrew Gelman and Cosma~Rohilla Shalizi.
\newblock Philosophy and the practice of {B}ayesian statistics.
\newblock \emph{British Journal of Mathematical and Statistical Psychology},
  66\penalty0 (1):\penalty0 8--38, 2013.

\bibitem[Hennig et~al.(2015)Hennig, Osborne, and Girolami]{Hennig2015}
Philipp Hennig, Michael~A Osborne, and Mark Girolami.
\newblock Probabilistic numerics and uncertainty in computations.
\newblock \emph{Proceedings of the Royal Society A}, 471\penalty0
  (2179):\penalty0 20150142, 2015.

\bibitem[Kersting and Hennig(2016)]{Kersting2016}
Hans Kersting and Philipp Hennig.
\newblock Active uncertainty calibration in {B}ayesian {ODE} solvers.
\newblock \emph{arXiv preprint:1605.03364}, 2016.

\bibitem[Larkin(1972)]{Larkin1972}
FM~Larkin.
\newblock Gaussian measure in {H}ilbert space and applications in numerical
  analysis.
\newblock \emph{The Rocky Mountain Journal of Mathematics}, 2:\penalty0
  379--421, 1972.

\bibitem[L{\'o}pez-Lopera et~al.(2017)L{\'o}pez-Lopera, Bachoc, Durrande, and
  Roustant]{Lopez-Lopera2017}
Andr{\'e}s~F L{\'o}pez-Lopera, Fran{\c{c}}ois Bachoc, Nicolas Durrande, and
  Olivier Roustant.
\newblock Finite-dimensional {G}aussian approximation with linear inequality
  constraints.
\newblock \emph{arXiv preprint:1710.07453}, 2017.

\bibitem[Perilla et~al.(2015)Perilla, Goh, Cassidy, Liu, Bernardi, Rudack, Yu,
  Wu, and Schulten]{Perilla2015}
Juan~R Perilla, Boon~Chong Goh, C~Keith Cassidy, Bo~Liu, Rafael~C Bernardi,
  Till Rudack, Hang Yu, Zhe Wu, and Klaus Schulten.
\newblock Molecular dynamics simulations of large macromolecular complexes.
\newblock \emph{Current Opinion in Structural Biology}, 31:\penalty0 64--74,
  2015.

\bibitem[Schober et~al.(2014)Schober, Duvenaud, and Hennig]{Schober2014}
Michael Schober, David~K Duvenaud, and Philipp Hennig.
\newblock Probabilistic {ODE} solvers with {R}unge-{K}utta means.
\newblock In \emph{Advances in Neural Information Processing Systems}, pages
  739--747, 2014.

\bibitem[Schober et~al.(2016)Schober, S{\"a}rkk{\"a}, and Hennig]{Schober2016}
Michael Schober, Simo S{\"a}rkk{\"a}, and Philipp Hennig.
\newblock A probabilistic model for the numerical solution of initial value
  problems.
\newblock \emph{Statistics and Computing}, pages 1--24, 2016.

\bibitem[Skilling(1992)]{Skilling1992}
John Skilling.
\newblock Bayesian solution of ordinary differential equations.
\newblock In \emph{Maximum Entropy and Bayesian Methods}, pages 23--37.
  Springer, 1992.

\bibitem[Teymur et~al.(2016)Teymur, Zygalakis, and Calderhead]{Teymur2016}
Onur Teymur, Kostas Zygalakis, and Ben Calderhead.
\newblock Probabilistic linear multistep methods.
\newblock In \emph{Advances in Neural Information Processing Systems}, pages
  4321--4328, 2016.

\bibitem[Teymur et~al.(2018)Teymur, Lie, Sullivan, and Calderhead]{Teymur2018}
Onur Teymur, Han~Cheng Lie, Tim Sullivan, and Ben Calderhead.
\newblock Implicit probabilistic integrators for {ODE}s.
\newblock \emph{arXiv preprint 1805.07970}, 2018.

\bibitem[Wedi(2014)]{Wedi2014}
Nils~P Wedi.
\newblock Increasing horizontal resolution in numerical weather prediction and
  climate simulations: illusion or panacea?
\newblock \emph{Philosophical Transactions of the Royal Society A},
  372\penalty0 (2018):\penalty0 20130289, 2014.

\end{thebibliography}

\appendix

\section{Appendix}

This appendix contains those definitions and concepts from the theory of Lie groups that are required to understand and implement our proposed Bayesian PNM.
In particular, a constructive proof of the key enabling result, Theorem \ref{1stordertoquad} in the main text, is provided.

\subsection{One-Parameter Lie Groups of Transformations}
\label{subsec: lie transformations}

\begin{definition}[One-Parameter Group of Transformations]
A \emph{one-parameter group of transformations} on a domain $D$ is a map
\begin{eqnarray*}
X : D \times S & \rightarrow & D \\
(x,\epsilon) & \mapsto & X(x,\epsilon) ,
\end{eqnarray*}
defined on $D \times S$ for some set $S \subset \mathbb{R}$, together with a bivariate map
\begin{eqnarray*}
\phi : S \times S & \rightarrow & S \\
(\epsilon,\delta) & \mapsto & \phi(\epsilon,\delta)
\end{eqnarray*}
such that the following hold:
\begin{enumerate}
\item[(1)] For each $\epsilon \in S$, the transformation $X(\cdot , \epsilon)$ is a bijection on $D$.
\item[(2)] $(S , \phi)$ forms a group with law of composition $\phi$.
\item[(3)] If $\epsilon_0$ is the identity element in $(\textit{S}, \phi)$, then $X(\cdot,\epsilon_0)$ is the identity map on $D$.
\item[(4)] For all $x \in D$ and $\epsilon,\delta \in S$, if $x^*=X(x,\epsilon), x^{**}= X(x^*,\delta)$, then $x^{**}=X(x^*,\phi (\epsilon,\delta))$.
\end{enumerate}
\end{definition}

In what follows we continue to use the shorthand notation $x^* = X(x,\epsilon)$.
The notion of a \emph{Lie} group additionally includes smoothness assumptions on the maps that constitute a group of transformations.
Recall that a real-valued function is \emph{analytic} if it can be locally expressed as a convergent power series.

\begin{definition}[One-Parameter Lie Group of Transformations]
Let $X$, together with $\phi$, form a one-parameter group of transformations on a domain $D$.
Then we say that $X$, together with $\phi$, form a \emph{one-parameter Lie group of transformations} on $D$ if, in addition, the following hold:
\begin{enumerate}
\item[(5)] $S$ is a (possibly unbounded) interval in $\mathbb{R}$. 
\item[(6)] For each $\epsilon \in S$, $X(\cdot,\epsilon)$ is infinitely differentiable in $\textit{D}$.
\item[(7)] For each $x \in D$, $X(x,\cdot)$ is an analytic function on $S$.
\item[(8)] $\phi$ is analytic in $S \times S$. 
\end{enumerate}
\end{definition}
\noindent Without loss of generality it will be assumed, through re-parametrisation if required, that $S$ contains the origin and $\epsilon=0$ is the identity element in $(S, \phi)$.
The definition is now briefly illustrated:

\begin{example}[Rotation Group] \label{ex: rotation}
The one-parameter transformation
\begin{eqnarray*}
x_1^* & = & x_1\cos(\epsilon)-x_2\sin(\epsilon) \\
x_2^* & = & x_1\sin(\epsilon)+x_2\cos(\epsilon)
\end{eqnarray*}
for $\epsilon \in \mathbb{R}$ forms a Lie group of transformations on $D = \mathbb{R}^2$ with group composition law $\phi(\epsilon,\delta)=\epsilon+\delta$.
\end{example}

\begin{definition}[Infinitesimal Transformation]
Let $X$ be a one-parameter Lie group of transformations. 
Then the transformation
\begin{eqnarray*}
x^* & = & x + \epsilon \xi(x) \\
\xi(x) & = & \left. \frac{\partial X(x,\epsilon)}{\partial \epsilon} \right|_{\epsilon = 0}
\end{eqnarray*}
is called the \emph{infinitesimal transformation} associated to $X$ and the map $\xi$ is called an \emph{infinitesimal}.
\end{definition}

\begin{definition}[Infinitesimal Generator]
The \emph{infinitesimal generator} of a one-parameter Lie group of transformations $X$ is defined to be the operator:
\begin{eqnarray*}
\mathrm{X} & = & \xi \cdot \nabla
\end{eqnarray*}
where $\xi$ is the infinitesimal associated to $X$ and $\nabla=(\frac{\partial}{\partial x_1}, \frac{\partial}{\partial x_2}, \dots , \frac{\partial}{\partial x_n})$ is the gradient.
\end{definition}

\begin{example}[Ex. \ref{ex: rotation}, continued]
For Ex. \ref{ex: rotation}, we have 
\begin{eqnarray*}
\xi(x) & = & \left( \left. \frac{\mathrm{d}x_1^*}{\mathrm{d}\epsilon} \right|_{\epsilon = 0}, \left. \frac{ \mathrm{d} x_2^*}{\mathrm{d}\epsilon} \right|_{\epsilon = 0}\right) \\
& = & (-x_2,x_1)
\end{eqnarray*}
so the infinitesimal generator for the rotation group is:
\begin{eqnarray*}
\mathrm{X} & = & -x_2\frac{\partial}{\partial x_1}+x_1\frac{\partial}{\partial x_2}
\end{eqnarray*}
\end{example}

The first fundamental theorem of Lie provides a constructive route to obtain the infinitesimal generator from the transformation itself:

\begin{theorem}[First Fundamental Theorem of Lie]
\label{thm: FTL 1}
A one parameter Lie group of transformations $X$ is characterised by the initial value problem:
\begin{eqnarray}
\frac{\mathrm{d}x^*}{\mathrm{d}\tau} & = & \xi(x^*) \label{eq: FTL 1} \\
x^* & = & x \hspace{10pt} \text{when } \tau = 0 
\end{eqnarray}
where $\tau(\epsilon)$ is a parametrisation of $\epsilon$ which satisfies $\tau(0) = 0$ and, for $\epsilon \neq 0$,
\begin{eqnarray}\label{eq:1stLietheorem}
\tau(\epsilon) & = & \int_0^{\epsilon} \cfrac{\partial \phi(a,b)}{\partial b} \Bigr|_{\substack{(a,b)=(\delta^{-1},\delta)}} \mathrm{d} \delta .
\end{eqnarray}
Here $\delta^{-1}$ denotes the group inverse element for $\delta$.
\end{theorem}
\begin{proof}
See pages 39-40 of \cite{Bluman2002}.
\end{proof}

Since Eq.~\eqref{eq: FTL 1} is translation-invariant in $\tau$, it follows that without loss of generality we can assume a parametrisation $\tau(\epsilon)$ such that the group action becomes $\phi(\tau_1,\tau_2) = \tau_1 + \tau_2$ and, in particular, $\tau^{-1} = - \tau$.
In the remainder, for convenience, we assume without loss of generality that all Lie groups are parametrised such that the group action is $\phi(\epsilon_1,\epsilon_2)  = \epsilon_1 + \epsilon_2$.

The next result can be viewed as a converse to Theorem \ref{thm: FTL 1}, as it shows how to obtain the transformation from the infinitesimal generator:

\begin{theorem} \label{thm: exp generator}
A one parameter Lie group of transformations with infinitesimal generator $\mathrm{X}$ is equivalent to
\begin{eqnarray*}
x^* & = & e^{{\epsilon}\mathrm{X}} x, \quad \mathrm{where} \quad e^{{\epsilon}\mathrm{X}}=\sum_{k=0}^{\infty} \frac{1}{k!}{\epsilon}^k\mathrm{X}^k x.
\end{eqnarray*}
\end{theorem}
\begin{proof}
From Taylor's theorem we have that
\begin{eqnarray*}
x^* & = & X(x,\epsilon) \\
& = & \sum_{k=0}^\infty \frac{\epsilon^k}{k!} \left. \frac{\partial^k X(x,\epsilon)}{\partial \epsilon^k} \right|_{\epsilon = 0}
\end{eqnarray*}
For any differentiable function $F$ we have that
\begin{eqnarray*}
\frac{\mathrm{d}F(x^*)}{\mathrm{d}\epsilon} \; = \; \sum_{i=1}^d \frac{\partial F(x^*)}{\partial x_i^*} \frac{\mathrm{d}x_i^*}{\mathrm{d}\epsilon} \; = \; \sum_{i=1}^d \xi_i \frac{\partial F(x^*)}{\partial x_i^*} \; = \; \mathrm{X} F(x^*)
\end{eqnarray*}
and similarly
\begin{eqnarray*}
\frac{\mathrm{d}^kF(x^*)}{\mathrm{d}\epsilon^k} & = & \mathrm{X}^k F(x^*).
\end{eqnarray*}
Thus
\begin{eqnarray*}
\left. \frac{\partial^k X(x,\epsilon)}{\partial \epsilon^k} \right|_{\epsilon = 0} & = & \mathrm{X}^k x
\end{eqnarray*}
so that the stated result is recovered.
\end{proof}

The following is immediate from the proof of Theorem \ref{thm: exp generator}:

\begin{corollary} \label{cor: exp}
If $F$ is infinitely differentiable, then $F(x^*) = e^{\epsilon \mathrm{X}} F(x)$.
\end{corollary}

\subsection{Invariance Under Transformation}
\label{subsec: invar under trans}

In this section we explain what it means for a curve or a surface to be invariant under a Lie group of transformations and how this notion relates to the infinitesimal generator.

\begin{definition}[Invariant Function]
A function $F:D \rightarrow \mathbb{R}$ is said to be \emph{invariant} under a one parameter Lie group of transformations $x^* = X(x,\epsilon)$ if $F(x^*)=F(x)$ for all $x \in D$ and $\epsilon \in S$.
\end{definition}

Based on the results in Section \ref{subsec: lie transformations}, one might expect that invariance to a transformation can be expressed in terms of the infinitesimal generator of the transformation.
This is indeed the case:

\begin{theorem}\label{thm: invar 1}
A differentiable function $F:D \mapsto \mathbb{R}$ is invariant under a one parameter Lie group of transformations with infinitesimal generator $\mathrm{X}$ if and only if $\mathrm{X}F(x)=0$ for all $x \in D$.
\end{theorem}
\begin{proof}
The result is established as follows:
\begin{eqnarray*}
F \text{ invariant} & \Leftrightarrow & F(x^*) = 0 \text{ whenever } F(x) = 0 \\
& \Leftrightarrow & e^{\epsilon \mathrm{X}}F(x) = 0 \text{ whenever } F(x) = 0 \quad \text{(Cor. \ref{cor: exp})} \\
& \Leftrightarrow & F(x) + \epsilon \mathrm{X} F(x) + O(\epsilon^2) = 0 \text{ whenever } F(x) = 0 \quad \text{(Taylor)} \\
& \Leftrightarrow & \mathrm{X} F(x) = 0 \text{ whenever } F(x) = 0
\end{eqnarray*} 
where the last line follows since the coefficient of the $O(\epsilon)$ term in the Taylor expansion must vanish.
This completes the proof.
\end{proof}

\begin{theorem} \label{thm: invar 2}
For a function $F:D \mapsto \mathbb{R}$ and a one parameter Lie group of transformations $x^* = X(x,\epsilon)$, the relation $F(x^*) = F(x)+\epsilon$ holds for all $x \in D$ and $\epsilon \in S$ if and only if $\mathrm{X}F(x)=1$ for all $x \in D$.
\end{theorem}
\begin{proof}
From Cor. \ref{cor: exp}, we have that $F(x^*) = F(x) + \epsilon \mathrm{X} F(x) + O(\epsilon^2)$.
The result follows from inspection of the $\epsilon$ coefficient.
\end{proof}

The following definition is fundamental to the method proposed in Section \ref{sec: methods} of the main text:

\begin{definition}[Canonical Coordinates] \label{def: canonical}
Consider a coordinate system $r = (r_1(x),  \dots , r_n(x))$ on a domain $D$.
Then any one parameter Lie group of transformations $x^* = X(x,\epsilon)$ induces a transformation of the coordinates $r_i^* = r_i(x^*)$.
The coordinate system $r$ is called \emph{canonical} for the transformation if:
\begin{eqnarray*}
r^*_1 & = & r_1  \\
& \vdots & \\
r^*_{n-1} & = & r_{n-1}  \\
r^*_n & = & r_n + \epsilon .
\end{eqnarray*}
\end{definition}

\begin{example}[Ex. \ref{ex: rotation}, continued] 
For the rotation group in Ex. \ref{ex: rotation}, we have canonical coordinates $r_1(x_1,x_2)=\sqrt{x_1^2+x_2^2}$ , $r_2(x_1,x_2)=\mathrm{arctan}(x_2/x_1)$.
\end{example}

In canonical coordinates, a one parameter Lie group of transformations can be viewed as a straight-forward translation in the $r_n$-axis.
The existence of canonical coordinates is established in Thm. 2.3.5-2 of \cite{Bluman2002}.
Note that Thms. \ref{thm: invar 1} and \ref{thm: invar 2} imply that $\mathrm{X}r^*_i=0$ for $i=1,2,...,n-1$, $\mathrm{X}r^*_n=1$.

\begin{definition}[Invariant Surface]
For a function $F:D \rightarrow \mathbb{R}$, a surface defined by $F(x)=0$ is said to be \emph{invariant} under a one parameter Lie group of transformation $x^* = X(x,\epsilon)$ if and only if $F(x^*)=0$ whenever $F(x)=0$ for all $x \in D$ and $\epsilon \in S$.
\end{definition}

The invariance of a surface, as for a function, can be cast in terms of an infinitesimal generator:

\begin{corollary} \label{thm: invar generator}
A surface $F(x)=0$ is invariant under a one parameter Lie group of transformations with infinitesimal generator $\mathrm{X}$ if and only if $\mathrm{X}F(x)=0$ whenever $F(x)=0$.
\end{corollary}

This is sufficient background on groups of transformations; next we turn to the use of Lie group methods in the ODE context.

\subsection{Symmetry Methods for ODEs}
\label{subsec: symm method ODE}

The aim of this section is to relate abstract Lie transformations to those ODEs for which these transformations are admitted.
These techniques form the basis for our proposed method in Section \ref{sec: methods}.

For an ODE of the form in Eq.~\eqref{eq: ODE1}, one can consider the action of a transformation on the coordinates $(x,y)$; i.e. a special case of the above framework where the generic coordinates $x_1$ and $x_2$ are respectively the independent ($x$) and dependent ($y$) variables of the ODE.
It is clear that such a transformation also implies some kind of transformation of the derivative $\mathrm{d}y / \mathrm{d}x$; this is made explicit next.
Note that in this paper we focus on ODEs, as opposed to PDEs, and as such the independent variable $x$ is scalar.
The notation 
\begin{eqnarray*}
y_m & := & \frac{\mathrm{d}^m y}{\mathrm{d} x^m}
\end{eqnarray*}
will be used.
Consider a one-parameter Lie group of transformations 
\begin{eqnarray*}
(x^*,y^*) = (X(x,y;\epsilon), Y(x,y;\epsilon)) .
\end{eqnarray*}
Then we have from the chain rule that
\begin{eqnarray*}
y^*_m & := & \frac{\mathrm{d}^m y^*}{\mathrm{d}(x^*)^m} 
\end{eqnarray*}
is a function of $x,y,y_1, \dots ,y_m$ and we denote $y_m^* = Y_m(x,y,y_1, \dots ,y_m;\epsilon)$.
As an explicit example:
\begin{eqnarray*}
y^*_1 & = & \frac{\mathrm{d}y^*}{\mathrm{d}x^*} \\
& = & \frac{\frac{\partial{Y(x,y;\epsilon)}}{\partial x}+y_1\frac{\partial{Y(x,y;\epsilon)}}{\partial y}}{\frac{\partial{X(x,y;\epsilon)}}{\partial x}+y_1\frac{\partial{X(x,y;\epsilon)}}{\partial y}} \; =: \; Y_1(x,y,y_1;\epsilon)
\end{eqnarray*}
In general:
\begin{eqnarray*}
y^*_m & = & \frac{\frac{\partial{y^*_{m-1}}}{\partial x}+y_1\frac{\partial{y^*_{m-1}}}{\partial y}+y_2\cfrac{\partial{y^*_{m-1}}}{\partial y_1}+...+y_m\frac{\partial{y^*_{m-1}}}{\partial y_{m-1}}}{\frac{\partial{X(x,y;\epsilon)}}{\partial x}+y_1\frac{\partial{X(x,y;\epsilon)}}{\partial y}} \; =: \; Y_m(x,y,y_1, \dots ,y_m;\epsilon)
\end{eqnarray*}
In this sense a transformation defined on $(x,y)$ can be naturally extended to a transformation on $(x,y,y_1,y_2,\dots)$ as required.

\begin{definition}[Admitted Transformation]
An $m$th order ODE $F(x,y,y_1, \dots , y_m)=0$ is said to \emph{admit} a one parameter Lie group of transformations $(x^*,y^*)=(X(x,y;\epsilon), Y(x,y;\epsilon))$ if the surface $F$ defined by the ODE is invariant under the Lie group of transformations, i.e. if $F(x^*,y^*,y_1^*, \dots ,y_m^*)=0$ whenever $F(x,y,y_1, \dots , y_m) = 0$.
\end{definition}

\begin{example}
Clearly any ODE of the form $\frac{\mathrm{d}y}{\mathrm{d}x}=F(x)$ admits the transformation $(x^*,y^*)=(x, y+\epsilon)$.
\end{example}

Our next task is to understand how the infinitesimal generator of a transformation can be extended to act on derivatives $y_m$.

\begin{definition}[Extended Infinitesimal Transformation]\label{extendedeta}
The $m$th \emph{extended infinitesimals} of a one parameter Lie group of transformations $(x^*,y^*) = (X(x,y;\epsilon) , Y(x,y;\epsilon) )$ are defined as the functions $\xi, \eta, \eta^{(1)} , \dots , \eta^{(m)}$ for which the following equations hold:
\begin{align*}
x^* & =  X(x,y;\epsilon) & = & \; x+\epsilon\xi(x,y)+O(\epsilon^2) \\
y^* & = Y(x,y;\epsilon) & = & \; y+\epsilon\eta(x,y)+O(\epsilon^2) \\
y^*_1 & = Y_1(x,y,y_1;\epsilon) & = & \; y_1+\epsilon\eta^{(1)}(x,y,y_1)+O(\epsilon^2) \\
& \vdots & & \\
y^*_m & = Y_m(x,y,y_1, \dots ,y_m;\epsilon) & = & \; y_m+\epsilon\eta^{(m)}(x,y,y_1,y_2, \dots , y_m)+O(\epsilon^2)
\end{align*}
\end{definition}
\noindent It can be shown straightforwardly via induction that 
\begin{eqnarray}
\eta^{(m)}(x,y,y_1,y_2, \dots , y_m) & = & \frac{\mathrm{d}^m\eta}{\mathrm{d}x^m}-\sum_{k=0}^{m}\frac{m!}{(m-k)!k!}y_{m-k-1}\frac{\mathrm{d}^k\xi}{\mathrm{d}x^k} \label{eq: extended infinitesimals}
\end{eqnarray}
where $\frac{\mathrm{d}}{\mathrm{d}x}$ denotes the full derivative with respect to $x$, i.e. $\frac{\mathrm{d}}{\mathrm{d}x}=\frac{\partial}{\partial x}+y_1\frac{\partial}{\partial y}+\sum_{k=2}^{m+1}y_k\frac{\partial}{\partial y_{k-1}}$.
It follows that $\eta^{(m)}$ is a polynomial in $y_1,y_2,\dots,y_m$ with coefficients linear combinations of $\xi, \eta$ and their partial derivatives up to the $m$th order.

\begin{definition}[Extended Infinitesimal Generator]\label{extendedx}
The $m$th \emph{extended infinitesimal generator} is defined as
\begin{eqnarray*}
\mathrm{X}^{(m)} & = & \xi_m(x,y,y_1, \dots ,y_m)\cdot\nabla \\
& = & {\xi}(x,y)\frac{\partial}{\partial x}+{\eta}(x,y)\frac{\partial}{\partial y}+{\eta^{(1)}}(x,y)\frac{\partial}{\partial y_1} + \dots +{\eta^{(m)}}(x,y,y_1, \dots ,y_m)\frac{\partial}{\partial y_m} 
\end{eqnarray*}
where $\nabla=(\frac{\partial}{\partial x}, \frac{\partial}{\partial y}, \frac{\partial}{\partial y_1} , \dots , \frac{\partial}{\partial y_m})$ is the extended gradient.
\end{definition}

The following corollaries are central to the actual computation of the admitted Lie groups of an ODE.

\begin{corollary}\label{thm: invar 3}
A differentiable function $F:D_m \rightarrow \mathbb{R}$ where $D_m$ is the phase space containing elements of the form $(x,y,y_1,\dots,y_m)$, is invariant under a one parameter Lie group of transformations with an extended infinitesimal generator $\mathrm{X}^{(m)}$ if and only if $\mathrm{X}^{(m)}F(x,y,y_1,\dots,y_m)=0$ for all $(x,y,y_1,\dots,y_m) \in D_m$.
\end{corollary}

\begin{corollary}[Infinitesimal Criterion for Symmetries Admitted by an ODE] \label{thm: Invar criteria}
A one parameter Lie group of transformations is admitted by the $m$th order ODE $F(x,y,y_1,..,y_m)=0$ if and only if its extended infinitesimal generator $\mathrm{X}^{(m)}$ satisfies $\mathrm{X}^{(m)}F(x,y,y_1,\dots ,y_m)=0$ whenever $F(x,y,y_1,\dots,y_m)=0$.
\end{corollary}

The reader now has sufficient background to address the case of a first-order ODE discussed in this paper.
In particular, we can now present a proof of the key enabling result stated in the main text.

\begin{proof} [Proof of Theorem \ref{1stordertoquad}]
\label{1stordertoquadproof}
Let the infinitesimal generator associated with the Lie group of transformations be denoted $\mathrm{X}$.
From the remarks after Def. \ref{def: canonical}, we can obtain canonical coordinates by solving the pair of first order partial differential equations $\mathrm{X}r=0$, $\mathrm{X}s=1$. 
By the chain rule we have 
\begin{eqnarray*}
\frac{\mathrm{d}s}{\mathrm{d}r} & = & \frac{s_x+s_{y}y'}{r_x+r_{y}y'} \; =: \; G(r,s)
\end{eqnarray*}
From the definition of canonical coordinates, the Lie group of transformations is $r^*=r$, $s^*=s+\epsilon$ in the transformed coordinate system, so 
\begin{eqnarray*}
\cfrac{\mathrm{d}s^*}{\mathrm{d}r^*} = G(r^*,s^*) & \implies & \cfrac{\mathrm{d}s}{\mathrm{d}r}=G(r,s+\epsilon)
\end{eqnarray*}
for all $\epsilon$, which implies $G(r,s)=G(r)$ and thus Eq.~\eqref{1storder} becomes
\begin{eqnarray*}
\cfrac{\mathrm{d}\mathrm{s}}{\mathrm{d}r} & = & G(r)
\end{eqnarray*}
as claimed.
\end{proof}

\end{document}